\documentclass[12pt]{article}
\usepackage[T1]{fontenc}

\usepackage[a4paper,centering,body={15.27cm,22.50cm}]{geometry}

\usepackage[numbers,sort]{natbib}
\usepackage{graphicx}
\usepackage{amsmath}
\usepackage{amsfonts}
\usepackage{amssymb}
\usepackage{amsthm}
\newtheorem{theorem}{Theorem}[section]
\newtheorem{lemma}[theorem]{Lemma}

\newtheorem{corollary}[theorem]{Corollary}
\title{A generalized nonisospectral Camassa-Holm equation and its multipeakon solutions}

\author{ Xiang-Ke Chang\footnotemark[1] \footnotemark[2] , Xiao-Min Chen\footnotemark\footnotemark[1] \footnotemark[2] and Xing-Biao Hu\footnotemark[1]}

\begin{document}
\bibliographystyle{plain}

\date{}
\maketitle

\renewcommand{\thefootnote}{\fnsymbol{footnote}}
\footnotetext[1]{LSEC, Institute of Computational Mathematics and
  Scientific Engineering Computing, AMSS, Chinese Academy of Sciences,
  P.O.Box 2719, Beijing 100190, PR China.}
\footnotetext[2]{University of Chinese Academy of
  Sciences, Beijing, PR China.}

\begin{abstract}
Motivated by the paper (Beals, Sattinger and Szmigielski, Adv. Math. 154 (2000) 229--257), we propose an extension of the Camassa-Holm equation, which also admits the multipeakon solutions.  The novel aspect is that our approach is mainly based on classic determinant technique. Furthermore, the proposed equation is shown to possess a nonisospectral Lax pair.
\end{abstract}
\begin{quote}
\noindent {\bfseries Keywords:}
Nonisospectral Camassa-Holm equation; Multipeakon solutions; Determinant technique
\end{quote}

\section{Introduction}
The celebrated Camassa-Holm equation taking the form
\begin{equation}\label{equ:CH}
m_t+(um)_x+mu_x=0, \ \ \ \  2m=4u-u_{xx}
\end{equation}
was derived as a shallow water wave equation by Camassa and Holm \cite{camassa1993integrable}.
This equation firstly appeared in the work of Fuchssteiner and Fokas \cite{fuchssteiner1981symplectic} as an abstract bi-Hamiltonian equation with infinitely many conservation laws. But it attracted no special attention until its rediscovery by Camassa and Holm. The most attractive character is that it admits N-peakon solutions, which are also called peakons, taking the form:
\begin{equation}\label{sol:C-H_form}
  u(x,t)=\frac{1}{2}\sum_{j=1}^Nm_j(t)exp(-2|x-x_j(t)|).
\end{equation}
It is remarked that, for the Camassa-Holm equation, the peakons are proved to be stable \cite{constantin2000stability,constantin2002stability}.

The ``peakons" have been of great research interest. In the literature, besides the Camassa-Holm equation, there exist many integrable systems with peaked solutions such as the Degasperis-Procesi equation \cite{degasperis1999asymptotic}, Novikov equation \cite{hone2008integrable,novikov2009generalisations}, Geng-Xue equation \cite{geng2009extension}, Hunter-Saxton Equation \cite{hunter1991dynamics} and some generalizations of Camassa-Holm equation \cite{chen2006two,geng2011three} etc.. ( There are plenty of papers on this topic. It was not our purpose trying to be exhaustive. Thus, we beg indulgence for the numerous omissions it certainly contains.)

The explicit expressions of multipeakon solutions to several integrable systems have been obtained.  For example,
in \cite{beals1999multipeakons,beals2000multipeakons}, Beals, Sattinger and Szmigielski used inverse spectral method and the Stieltjes theorem on continued fractions to get multipeakon solutions to the Camassa-Holm equation. The closed form solution is expressed in terms of the orthogonal polynomials of the related classical moment problem. As for the cases of the Degasperis-Procesi equation, the Novikov equation and the Geng-Xue equation, please consult \cite{beals2001inverse,hone2009explicit,lundmark2003multi,lundmark2005degasperis,lundmark2013inverse}.

Motivated by the work of Beals, Sattinger and Szmigielski, we have managed to confirm the multipeakon solutions from another way--the determinant technique and succeeded. With the help of determinant identities, we propose an extended Camassa-Holm equation, which also admits N-peakon solutions.

The proposed equation is
\begin{equation}\label{equ:GNCH}
  m_t+[((s+r)u-r\int_x^\infty u(y,t)d y)m]_x+m[(s+r)u_x+ru]=0, \ \ \ \  2m=4u-u_{xx},
\end{equation}
where $r,s\in\mathbb{C}$. Obviously, this equation reduces to the Camassa-Holm equation when $s=1,r=0$. Unexpectedly, we find that this equation possesses a nonisospectral Lax pair, in other words, the spectrum in the Lax pair is dependent on time $t$ instead of a constant. Therefore, we shall call it generalized nonisospectral Camassa-Holm (GNCH) equation.

As for the derivation of the GNCH equation, we have made an inverse calculation. Assume that the form of the explicit solution is still the same as that of the Camassa-Holm equation. We firstly alter the evolution with respect to time $t$ for the moments of Hankel determinant. Then we deduce the dynamical system by using the determinant identities. At last, the corresponding partial differential equation is obtained. To the best of our knowledge, this equation is novel.

The paper is organized as follows. In Section 2, the work of Beals et al. is reviewed. As our results are obtained by use of determinant technique, we show some determinant identities in Section 3. The explicit formulas of N-peakon solutions to the GNCH equation \eqref{equ:GNCH} are presented in Section 4 and the Lax pair is also given there. In section 5, we consider three special cases of the GNCH equation \eqref{equ:GNCH}. Section 6 is devoted to discussions.

\section{Review of the work by Beals et al.}
In \cite{beals1999multipeakons,beals2000multipeakons}, in order to obtain the explicit N-peakon solutions to the Camassa-Holm equation, Beals et al. used inverse spectral method, the theory of continued fractions and formulas of Stieltjes. Let us sketch the idea.

They considered the finite dimension case of the continuum equations \eqref{equ:CH} when $m$ is taken to be a discrete measure with weights $m_j$ at locations $x_j$:
\begin{equation}\label{measure:m}
  mdx=\sum_{j=1}^{N}m_j \delta(x-x_j)dx, \ \ \ x_1<x_2<\cdots <x_N.
\end{equation}
It is clear that the form \eqref{sol:C-H_form}
\[
u(x,t)=\frac{1}{2}\sum_{j=1}^Nm_j(t)exp(-2|x-x_j(t)|)
\]
can be calculated from \eqref{measure:m} and the second equation in \eqref{equ:CH}.

In this case, the Camassa-Holm equation may be written as a Hamiltonian system for $m_j,x_j$, which is  the same as that in \cite{camassa1993integrable} :
\begin{equation}\label{equ:CH-Hamilton}
 \dot{x}_j=\frac{\partial H}{\partial {m}_j}=u(x_j), \ \ \ \ \dot{m}_j=-\frac{\partial H}{\partial {x}_j}=-<u_x(x_j)>m_j,
\end{equation}
where
\begin{eqnarray*}
  H(x,m)&=&\int_{-\infty}^{\infty}(u^2+\frac{1}{4}u_x^2)\ dx=\frac{1}{4}\sum_{j,k=1}^Nm_jm_ke^{-2|x_j-x_k|}\\
  &=&\frac{1}{2}\sum_{j=1}^Nu(x_j)m_j=\frac{1}{2}\int_{-\infty}^{\infty}u(x)dm(x)
\end{eqnarray*}
and the notation $<f(x_j)>$ denotes
\[<f(x_j)>=\frac{f(x_j+)+f(x_j-)}{2}\]
at the jump point $x_j$ of $f(x)$.

The amplitudes $m_j$ and the locations $x_j$ of the peaks were calculated explicitly by using the inverse scattering approach. The spectral data for this spectral problem consists of the eigenvalues $\{\lambda_j\}$ and the coupling coefficients $\{c_j\}$.

It is known that the Camassa-Holm equation \eqref{equ:CH} possesses the Lax pair \cite{camassa1993integrable}:
\begin{eqnarray*}
  L_0(z)f(x,t)=0,&\ \ \ \ &L_0(z)=\partial_x^2-zm(x,t)-1,\\
  \frac{\partial f(x,t)}{\partial t}=A_0(z)f,&\ \ \ \ & A_0(z)=(\frac{1}{z}-u(x,t))\partial_x+\frac{1}{2}u_x(x,t),
\end{eqnarray*}
where $z$ is a constant. From the above equation, they found that the time evolution for the eigenfunctions $b_j$ and the coupling coefficients satisfy
\[\frac{\partial b_j}{\partial t}=0,\ \ \ \ \frac{\partial c_j}{\partial t}=-\frac{2c_j}{\lambda_j}.
\]

In order to make the problem simpler, they mapped the spectral problem into
\begin{equation*}
  \tilde{f}_{yy}=zg(y)\tilde{f},\ \ \ -1<y<1;\ \ \ \tilde{f}(-1)=\tilde{f}(1)=0
\end{equation*}
with  $g(y)=\frac{m(x)}{(1-y^2)^2}$, by using the Liouville transformation
\begin{equation*}
 y=\tanh x, \ \ \ \ f(x)=\frac{\tilde{f}(y)}{\sqrt{(1-y^2)}}.
\end{equation*}
In this case, the eigenvalues and coupling coefficients are preserved. And the measure $m$ is transformed into
\begin{equation}
  gdy=\sum_{j=1}^{N}g_j \delta(y-y_j)dy, \ \ \ -1=y_0<y_1<y_2<\cdots <y_N<y_{N+1}=1,
\end{equation}
where
\begin{equation}\label{liouville_tran1}
  y_j=\tanh x_j, \ \ \ \ g_j=\frac{m_j}{1-y_j^2}.
\end{equation}

Then they solved the inverse spectral problem for the finite string on $(-1,1)$ and they obtained
\begin{equation}
  y_j=1-\frac{\Delta_{N-j}^2}{\Delta_{N-j+1}^0}, \ \ \ \ g_j=\frac{(\Delta_{N-j+1}^0)^2}{\Delta_{N-j+1}^1\Delta_{N-j}^1}
\end{equation}
where $\Delta_k^l$ denotes the determinant $\det(A_{i+j+l}(t))_{i,j=0}^{k-1}$ with the convention $\Delta_0^l=1$ and $\Delta_k^l=0$ for $k<0$. And here $A_k(t)$ are defined by
\begin{equation}\label{rela:A-a}
  A_k(t)=\sum_{j=0}^{N}(-\lambda_j)^ka_j(t)
\end{equation}
with
\begin{equation}\label{rela:a}
  a_0=\frac{1}{2}, \ \ \ \dot{a}_j=-\frac{2a_j}{\lambda_j}, j\geq1
\end{equation}
and the spectrum $\{\lambda_j\}$ satisfying $\lambda_0=0$. Here we mention that we shall use the convention $0^0=1$ in the full text.

Thus the explicit formulas of N-peakon solutions were obtained by combining \eqref{liouville_tran1}-\eqref{rela:a}. We summarize the result as follows.
\begin{theorem}[Beals, Sattinger and Szmigielski]\label{th:beals}
  The Camassa-Holm equation \eqref{equ:CH} admits the N-peakon solution \[
u(x,t)=\frac{1}{2}\sum_{j=1}^Nm_j(t)\exp(-2|x-x_j(t)|)
\]
 with
\begin{equation}
    x_j=\frac{1}{2}\log\left(\frac{1+y_j}{1-y_j}\right), \ \ \ \ m_j=g_j(1-y_j^2)
  \end{equation}
  and
  \begin{equation}
    y_j=1-\frac{\Delta_{N-j}^2}{\Delta_{N-j+1}^0}, \ \ \ \ g_j=\frac{(\Delta_{N-j+1}^0)^2}{\Delta_{N-j+1}^1\Delta_{N-j}^1}.
  \end{equation}
  Here $\Delta_k^l=\det(A_{i+j+l}(t))_{i,j=0}^{k-1}$ and the moments $A_k(t)$ are restricted by
  \begin{eqnarray}
      A_k(t)=\sum_{j=0}^{N}(-\lambda_j)^ka_j(t),\ \ \ a_0(t)=\frac{1}{2},\ \ \ \lambda_0=0
  \end{eqnarray}
  with $\dot{a}_j(t)=-\frac{2a_j(t)}{\lambda_j}, j\geq1$ and nonzero real constants $\lambda_j, j\geq1$.
\end{theorem}

\textbf{Remark:} In order to make $x_j$ and $m_j$ exist, $\lambda_j$ are required to be distinct, $a_j(0)$ are positive, and the determinant $\Delta_{k}^1$ for $1\leq k\leq N$ do not vanish. And one sufficient condition for $\Delta_{k}^1\neq0,\ 1\leq k\leq N$ is that all the $\lambda_j$ have the same sign.

\section{Determinant identities}
Noting that the explicit formulae for the N-peakon solutions to the Camassa-Holm equation can be expressed by Hankel determinants, it is necessary to seek for some potential properties behind them.  In this section, we shall present some interesting results for the Hankel determinants, some of which have appeared in \cite{beals2000multipeakons}.

The object in this section is the Hankel determinant behaving as $\Gamma_k^l=\det(B_{i+j+l})_{i,j=0}^{k-1}$ with the convention $\Gamma_0^l=1$ and $\Gamma_k^l=0$ for $k<0$.
  Besides, we also introduce the determinant taking the form of
  \[F_k^l=
  \left|\begin{array}{ccccc}
    B_l&B_{l+2}&B_{l+3}&\cdots &B_{l+k}\\
     B_{l+1}&B_{l+3}&B_{l+4}&\cdots &B_{l+k+1}\\
    \vdots&\vdots&\vdots&\ddots&\vdots\\
    B_{l+k-1}&B_{l+k+1}&B_{l+k+2}&\cdots&B_{l+2k-1}
    \end{array}\right|
  \]
  with the convention $F_k^l=0,k\leq0$.
\subsection{Linear identities}
\begin{lemma}\label{lemma:linear}
If the elements $B_k, k\in \mathbb{Z}$  subject to
  \begin{eqnarray*}
      B_k=\sum_{j=1}^{N}(\mu_j)^kb_j, \ \ b_j\neq0,\ \
      \mu_j\neq0,
  \end{eqnarray*}
  then there hold
  \begin{eqnarray*}
    &&\Gamma_N^l=\prod_{i=1}^N[b_i(\mu_i)^l]\prod_{1\leq i<j\leq N}(\mu_j-\mu_i)^2,\\
    &&\Gamma_k^l=0,\ \ k>N
  \end{eqnarray*}
for any $l\in \mathbb{Z}$.
\end{lemma}
\begin{proof}
  Noting that any $B_k$ is a linear combination of $b_j, j=1,2,\cdots,N$, the matrix $(B_{i+j+l})_{i,j=0}^{k-1}$ may be written as a product of two matrices, namely,
  \begin{eqnarray*}
    &&\left(\begin{array}{cccc}
    B_l&B_{l+1}&\cdots &B_{l+k-1}\\
     B_{l+1}&B_{l+2}&\cdots &B_{l+k}\\
    \vdots&\vdots&\ddots&\vdots\\
    B_{l+k-1}&B_{l+k}&\cdots&B_{l+2k-2}
    \end{array}\right)\\
    &=&\left(\begin{array}{cccc}
    1&1&\cdots &1\\
     \mu_1&\mu_2&\cdots &\mu_N\\
    \vdots&\vdots&\ddots&\vdots\\
    \mu_1^{k-1}&\mu_2^{k-1}&\cdots&\mu_N^{k-1}
    \end{array}\right)_{k\times N}\\
    &&\cdot
    \left(\begin{array}{cccc}
     b_1(\mu_1)^l&b_1(\mu_1)^{l+1}&\cdots &b_1(\mu_1)^{l+k-1}\\
     b_2(\mu_2)^l&b_2(\mu_2)^{l+1}&\cdots &b_2(\mu_2)^{l+k-1}\\
    \vdots&\vdots&\ddots&\vdots\\
    b_N(\mu_N)^l&b_N(\mu_N)^{l+1}&\cdots &b_N(\mu_N)^{l+k-1}
    \end{array}\right)_{N\times k}.
  \end{eqnarray*}
  Based on the above decomposition, we see that the rank of the matrix $(B_{i+j+l})_{i,j=0}^{k-1}$ is not more than $N$. Thus the determinant $\Gamma_k^l$ is equal to zero if $k>N$. When $k=N$, the result can be obtained by noting the Vandermonde matrices on the right.
\end{proof}
\begin{corollary}\label{coro:linear}
If the elements $B_k\in \mathbb{Z}$ subject to
  \begin{eqnarray*}
      B_k=\sum_{j=0}^{N}(\mu_j)^kb_j,
  \end{eqnarray*}
  with $b_0=\frac{1}{2},\mu_0=0$ and $b_j\neq0,\mu_j\neq0$ for $j\geq1$,  then there hold
    \begin{eqnarray*}
    &&\Gamma_N^0-\frac{1}{2}\Gamma_{N-1}^2=\prod_{i=1}^Nb_i\prod_{1\leq i<j\leq N}(\mu_j-\mu_i)^2,\\
    &&\Gamma_N^l=\prod_{i=1}^N[b_i(\mu_i)^l]\prod_{1\leq i<j\leq N}(\mu_j-\mu_i)^2,\ \  l\geq1,\\
    &&\Gamma_k^0-\frac{1}{2}\Gamma_{k-1}^2=0,\ \ k>N,\\
    &&\Gamma_k^l=0,\ \ k>N,\ \ l\geq1,\\
    &&\Gamma_{N+1}^{-1}+F_N^0-\frac{1}{4}\Gamma_{N-1}^0=0.
  \end{eqnarray*}
\end{corollary}
\begin{proof}
  If we let
  \[B_k^*=\left\{
  \begin{array}{ll}
    B_k, &k\neq0\\
    B_0-\frac{1}{2}, & k=0
  \end{array}\right.,
  \]
  then
  \[
  B_k^*=\sum_{j=1}^{N}(\mu_j)^kb_j.
  \]
  The five formulae are the consequences of applying Lemma \ref{lemma:linear} to different determinants with the elements $B_k^*$. The proofs can be achieved by noting that the left sides in the above formulae are $\det(B_{i+j}^*)_{i,j=0}^{N-1}$, $\det(B_{i+j+l}^*)_{i,j=0}^{N-1}$, $\det(B_{i+j}^*)_{i,j=0}^{k-1}$, $\det(B_{i+j+l}^*)_{i,j=0}^{k-1}$ and $\det(B_{i+j-1}^*)_{i,j=0}^{k-1}$, respectively.
\end{proof}

\subsection{Bilinear identities}
 For any determinant $D$, the well known Jacobi determinant identity \cite{aitken1959determinants,brualdi1983determinantal} reads
\begin{eqnarray*}
D D\left(\begin{array}{cc}
i_1 & i_2 \\
j_1 & j_2 \end{array}\right)&=&D\left(\begin{array}{c}
i_1  \\
j_1 \end{array}\right)D\left(\begin{array}{c}
i_2  \\
j_2 \end{array}\right)-D\left(\begin{array}{c}
i_1  \\
j_2 \end{array}\right)D\left(\begin{array}{c}
i_2  \\
j_1 \end{array}\right),
\end{eqnarray*}
where $D\left(\begin{array}{cccc}
i_1&i_2 &\cdots& i_k\\
j_1&j_2 &\cdots& j_k
\end{array}\right),\ i_1<i_2<\cdots<i_k,\ j_1<j_2<\cdots<j_k$ denotes the
determinant of the matrix obtained from $D$ by removing the rows with indices
$i_1,i_2 ,\cdots, i_k$ and the columns with indices $j_1,j_2,\cdots j_k$.
Applying the Jacobi determinant identity, we have the following identities.
\begin{lemma}\label{lemma:bilinear}
  For any $k\geq-1$, $l\in \mathbb{Z}$,
  \begin{eqnarray*}
    &&\Gamma_{k+2}^{l-1}\Gamma_{k}^{l+1}=\Gamma_{k+1}^{l-1}\Gamma_{k+1}^{l+1}-(\Gamma_{k+1}^l)^2,\\
    &&\Gamma_{k+1}^{l-1}\Gamma_{k}^{l+1}=F_{k+1}^{l-1}\Gamma_{k}^{l}-F_{k}^{l-1}\Gamma_{k+1}^{l},\\
    &&\Gamma_{k+1}^{l}F_{k}^{l}=\Gamma_{k}^{l+1}F_{k+1}^{l-1}-\Gamma_{k+1}^{l-1}\Gamma_{k}^{l+2},\\
    &&\Gamma_{k+2}^{l-1}\Gamma_{k}^{l+2}=\Gamma_{k+1}^{l+1}F_{k+1}^{l-1}-\Gamma_{k+1}^lF_{k+1}^l.
  \end{eqnarray*}
\end{lemma}
\begin{proof}
  The proofs for the case of $k=-1$ are obvious by noting the convention $\Gamma_0^l=1,\Gamma_k^l=0,k<0$ and $F_k^l=0,k\leq0$.

  Now we consider the case of $k\geq0$. Taking
  \[
  D_1=\left|\begin{array}{cccc}
    B_{l-1}&B_{l}&\cdots &B_{l+k}\\
     B_{l}&B_{l+1}&\cdots &B_{l+k+1}\\
    \vdots&\vdots&\ddots&\vdots\\
    B_{l+k}&B_{l+k+1}&\cdots&B_{l+2k+1}
    \end{array}\right|,
    \]
     \[
  D_2=\left|\begin{array}{ccccc}
    0&B_{l-1}&B_{l}&\cdots &B_{l+k-1}\\
    1&B_{l}&B_{l+1}&\cdots &B_{l+k}\\
    0&B_{l+1}&B_{l+2}&\cdots &B_{l+k+1}\\
    \vdots&\vdots&\vdots&\ddots&\vdots\\
    0&B_{l+k}&B_{l+k+1}&\cdots&B_{l+2k}
    \end{array}\right|,
    \]
    and setting
    \[i_1=j_1=1, i_2=j_2=k+2,\]
    the Jacobi identity yields to the first two equalities, respectively.

    The last two equalities are the consequences of applying the Jacobi identity to
    \[
  D_3=\left|\begin{array}{ccccc}
    0&B_{l-1}&B_{l}&\cdots &B_{l+k-1}\\
    1&B_{l}&B_{l+1}&\cdots &B_{l+k}\\
    0&B_{l+1}&B_{l+2}&\cdots &B_{l+k+1}\\
    \vdots&\vdots&\vdots&\ddots&\vdots\\
    0&B_{l+k}&B_{l+k+1}&\cdots&B_{l+2k}
    \end{array}\right|,
    \]
    and
    \[
    D_4=\left|\begin{array}{cccc}
    B_{l-1}&B_{l}&\cdots &B_{l+k}\\
     B_{l}&B_{l+1}&\cdots &B_{l+k+1}\\
    \vdots&\vdots&\ddots&\vdots\\
    B_{l+k}&B_{l+k+1}&\cdots&B_{l+2k+1}
    \end{array}\right|,
    \]
    with
    \[i_1=j_1=1, i_2=k+2, j_2=2,\] respectively.
\end{proof}
\subsection{Combinations of identities}
Applying Lemma \ref{lemma:bilinear}, the following corollaries are easily obtained.
\begin{corollary}\label{coro:sum}
For $0\leq k\leq N-1$,
  \begin{eqnarray*}
    &&\sum_{l=k+1}^{N-1}\frac{(\Gamma_{l+1}^0)^2}{\Gamma_{l+1}^1\Gamma_{l}^1}=\frac{\Gamma_{k+2}^{-1}}{\Gamma_{k+1}^1}-\frac{\Gamma_{N+1}^{-1}}{\Gamma_{N}^1},\\
    &&\sum_{l=k+1}^{N-1}\frac{\Gamma_{l+1}^0\Gamma_{l}^2}{\Gamma_{l+1}^1\Gamma_{l}^1}=\frac{F_{N}^{0}}{\Gamma_{N}^1}-\frac{F_{k+1}^{0}}{\Gamma_{k+1}^1},\\
    &&\sum_{l=k+1}^{N-1}\frac{(\Gamma_{l}^2)^2}{\Gamma_{l+1}^1\Gamma_{l}^1}=\frac{\Gamma_{N-1}^{3}}{\Gamma_{N}^1}-\frac{\Gamma_{k}^{3}}{\Gamma_{k+1}^1},\\
    &&\sum_{l=0}^{k}\frac{(\Gamma_{l}^2)^2}{\Gamma_{l+1}^1\Gamma_{l}^1}=\frac{\Gamma_{k}^{3}}{\Gamma_{k+1}^1}.
  \end{eqnarray*}
\end{corollary}
\begin{proof}
  We shall take the proof to the last equality for example and the others are omitted because the proofs are similar. The detail of the proof to the last one is
  \begin{eqnarray*}
  &&\sum_{l=0}^{k}\frac{(\Gamma_{l}^2)^2}{\Gamma_{l+1}^1\Gamma_{l}^1}\\
  &=&\sum_{l=0}^{k}\frac{\Gamma_{l}^1\Gamma_{l}^3-\Gamma_{l+1}^1\Gamma_{l-1}^3}{\Gamma_{l+1}^1\Gamma_{l}^1}\\
  &=&\sum_{l=0}^{k}[\frac{\Gamma_{l}^3}{\Gamma_{l+1}^1}-\frac{\Gamma_{l-1}^3}{\Gamma_{l}^1}]\\
  &=&\frac{\Gamma_{k}^{3}}{\Gamma_{k+1}^1},
  \end{eqnarray*}
  where we have used the convention $\Gamma_{-1}^3=0$.
\end{proof}
\begin{corollary}\label{coro:multiple-identity}
For $0\leq k\leq N-1$, there hold
  \begin{eqnarray}
  &&(F_{k+1}^{-1}+F_{k}^1)\Gamma_{k}^2-F_{k}^1\Gamma_{k+1}^0\nonumber\\
  &=&(\Gamma_{k}^2)^2(\frac{\Gamma_{k+2}^{-1}}{\Gamma_{k+1}^1}+\frac{F_{k+1}^{0}}{\Gamma_{k+1}^1})+((\Gamma_{k+1}^0)^2-\Gamma_{k+1}^0\Gamma_{k}^2)\frac{\Gamma_{k}^{3}}{\Gamma_{k+1}^1},\label{multi-identity1}\\
&&2\Gamma_{k+1}^1\Gamma_{k}^1(2F_{k+1}^{-1}+2F_k^1)-\Gamma_{k+1}^0[(2F_{k+1}^0-\Gamma_{k}^3)\Gamma_{k}^1+\Gamma_{k+1}^1(2F_{k}^0-\Gamma_{k-1}^3)]\nonumber\\
&=&\Gamma_{k}^1\Gamma_{k}^2(4\Gamma_{k+2}^{-1}+4F_{k+1}^0-\Gamma_{k}^{3})-\Gamma_{k+1}^1\Gamma_{k-1}^{3}(2\Gamma_{k+1}^0-\Gamma_{k}^2)\nonumber\\
&&+(\Gamma_{k+1}^0-\Gamma_{k}^2)[2\Gamma_{k+1}^0\Gamma_{k}^2-(\Gamma_{k}^2)^2].\label{multi-identity2}
\end{eqnarray}
\end{corollary}
\begin{proof}
The first equality can be confirmed by applying Lemma \ref{lemma:bilinear} and eliminating $F_k^{-1},F_k^0,F_k^1$. More exactly, the proofs can be completed by noting that
\begin{eqnarray*}
  &&F_{k+1}^{-1}\Gamma_{k}^2+F_{k}^1\Gamma_{k}^2-F_{k}^1\Gamma_{k+1}^0-(\Gamma_{k}^2)^2\frac{F_{k+1}^{0}}{\Gamma_{k+1}^1}\\
  &=&\frac{\Gamma_{k}^{2}}{\Gamma_{k+1}^1}(\Gamma_{k+1}^0F_{k+1}^0+\Gamma_{k+2}^{-1}\Gamma_{k}^2)+\frac{\Gamma_{k}^{2}}{\Gamma_{k+1}^1}(\Gamma_{k+1}^1F_{k}^1)\\
  &&-\frac{\Gamma_{k+1}^{0}}{\Gamma_{k+1}^1}(\Gamma_{k}^2F_{k+1}^0-\Gamma_{k+1}^0\Gamma_{k}^3)-\frac{\Gamma_{k}^{2}}{\Gamma_{k+1}^1}(\Gamma_{k+1}^1F_{k}^1+\Gamma_{k+1}^0\Gamma_{k}^3)\\
  &=&(\Gamma_{k}^2)^2\frac{\Gamma_{k+2}^{-1}}{\Gamma_{k+1}^1}+((\Gamma_{k+1}^0)^2-\Gamma_{k+1}^0\Gamma_{k}^2)\frac{\Gamma_{k}^{3}}{\Gamma_{k+1}^1}.
\end{eqnarray*}

The second equality may also be verified by employing Lemma \ref{lemma:bilinear} and eliminating $F_k^{-1},F_k^0,F_k^1,\Gamma_{k-1}^3$. More precisely, it is proved by noting that
\begin{eqnarray*}
&&2\Gamma_{k+1}^1\Gamma_{k}^1(2F_{k+1}^{-1}+2F_k^1)-\Gamma_{k+1}^0(2F_{k+1}^0\Gamma_{k}^1+2\Gamma_{k+1}^1F_{k}^0)-4\Gamma_{k}^1\Gamma_{k}^2F_{k+1}^0\\
&&+3\Gamma_{k+1}^0\Gamma_{k+1}^1\Gamma_{k-1}^3-\Gamma_{k+1}^1\Gamma_{k-1}^3\Gamma_{k}^2\\
&=&4\Gamma_{k}^1(\Gamma_{k+1}^1F_{k+1}^{-1}-\Gamma_{k+1}^0F_{k+1}^0)+4\Gamma_{k}^1(\Gamma_{k+1}^1F_k^1-\Gamma_{k}^2F_{k+1}^0)\\
&&+2\Gamma_{k+1}^0(F_{k+1}^0\Gamma_{k}^1-\Gamma_{k+1}^1F_{k}^0)+3\Gamma_{k+1}^0(\Gamma_{k}^1\Gamma_{k}^3-(\Gamma_{k}^2)^2)+\Gamma_{k}^2(\Gamma_{k}^1\Gamma_{k}^3-(\Gamma_{k}^2)^2)\\
&=&4\Gamma_{k}^1\Gamma_{k+2}^{-1}\Gamma_{k}^2-\Gamma_{k}^1\Gamma_{k+1}^0\Gamma_{k}^3+2\Gamma_{k+1}^0\Gamma_{k+1}^0\Gamma_{k}^2-3\Gamma_{k+1}^0(\Gamma_{k}^2)^2-\Gamma_{k}^1\Gamma_{k}^2\Gamma_{k}^3+(\Gamma_{k}^2)^3.
\end{eqnarray*}
\end{proof}

\section{The GNCH equation \eqref{equ:GNCH}}
We mention that Theorem \ref{th:beals}, which describes the N-peakon solutions to the Camassa-Holm equation,  can be proved by using determinant technique. Assume that the form of the N-peakon solution is the same as \eqref{sol:C-H_form}. By modifying the time evolution in \eqref{rela:a}, we obtain a generalized equation after some technical inverse operations, which is the GNCH equation \eqref{equ:GNCH}. Obviously, it also admits the N-peakon solution. The result is described as below.
\begin{theorem}\label{th:GNCH-solution}
The GNCH equation \eqref{equ:GNCH} admits the N-peakon solution having the form of
\begin{equation}\label{sol:GNCH_form}
  u(x,t)=\frac{1}{2}\sum_{j=1}^Nm_j(t)exp(-2|x-x_j(t)|).
\end{equation}
The explicit expressions for $m_j$ and $x_j$ are as follows:
\begin{equation}\label{expression:xm2}
    x_j=\frac{1}{2}\log\left(\frac{1+y_j}{1-y_j}\right), \ \ \ \ m_j=g_j(1-y_j^2),
  \end{equation}
  with
  \begin{equation}\label{expression:gy2}
    y_j=1-\frac{\Delta_{N-j}^2}{\Delta_{N-j+1}^0}, \ \ \ \ g_j=\frac{(\Delta_{N-j+1}^0)^2}{\Delta_{N-j+1}^1\Delta_{N-j}^1}.
  \end{equation}
 Here $\Delta_k^l=\det(A_{i+j+l}(t))_{i,j=0}^{k-1}$ and the moments $A_k(t)$ are restricted by
  \begin{eqnarray}\
      A_k(t)=\sum_{j=0}^{N}(-\lambda_j)^k e^{[r(k+1)+2s]a_j(t)},\ \ \ a_0(t)=\frac{\log\frac{1}{2}}{r+2s},\ \ \ \lambda_0=0
  \end{eqnarray}
  with $\dot{a}_j(t)=-\frac{1}{\lambda_je^{ra_j(t)}}, j\geq1$ and nonzero real constants $\lambda_j$.
\end{theorem}

 We shall present our result
by the way of proving Theorem \ref{th:GNCH-solution} rather than how to derive the GNCH equation \eqref{equ:GNCH}. As the Camassa-Holm equation \eqref{equ:CH} is the special case of the generalized nonisospectral Camassa-Holm equation \eqref{equ:GNCH}, it means that we give an alternative proof to the case of the Camassa-Holm equation.

In order to prove Theorem \ref{th:GNCH-solution}, we need some lemmas. We remark that Corollaries \ref{coro:linear},\ref{coro:sum},\ref{coro:multiple-identity} and Lemma \ref{lemma:bilinear} in Section 3 still hold in the case of replacing $\Gamma_k^l, F_k^l, b_k, \mu_k$ by $\Delta_k^l,G_k^l, e^{(r+2s)a_k}, -\lambda_ke^{ra_k}$, respectively.
Here $G_k^l$ is defined by
\[G_k^l=
  \left|\begin{array}{ccccc}
    A_l&A_{l+2}&A_{l+3}&\cdots &A_{l+k}\\
     A_{l+1}&A_{l+3}&A_4&\cdots &A_{l+k+1}\\
    \vdots&\vdots&\vdots&\ddots&\vdots\\
    A_{l+k-1}&A_{l+k+1}&A_{k+2}&\cdots&A_{l+2k-1}
    \end{array}\right|
  \]
with the convention $G_k^l=0,k\leq0$.

Besides, as for the derivative of $\Delta_k^l$ with respect to $t$, we also have the following properties.
\begin{lemma}\label{lemma:derivative2}
For $k\geq0$,
\begin{eqnarray}
  &&\dot{\Delta_k^0}=(r+2s)G_k^{-1}+(2r+2s)G_{k-1}^1,\\
  &&\dot{\Delta_k^1}=(2r+2s)G_k^{0}-(r+s)\Delta_{k-1}^3,\\
  &&\dot{\Delta_k^l}=[r(l+1)+2s]G_k^{l-1}, \ \ l\geq2.
\end{eqnarray}
\end{lemma}
\begin{proof}
Here we shall give a detailed proof of the first formula. The proofs
of the rest two can be achieved by following the steps of the proof of the first one.

 It is easy to see that the element $A_k(t)$ satisfy the differential equation
\begin{equation}\label{rela:A_t2}
  \dot{A}_k=[r(k+1)+2s]A_{k-1}, \ \ \ k=0 \ {\text{and}}\ k\geq2, \ \ \  \dot{A}_1=(2r+2s)A_{0}-(r+s),
\end{equation}
where we have used $\lambda_0=0$  and the convention $0^0=1$.

By using the differential rule for determinants, we have
\begin{eqnarray*}
  \dot{\Delta_k^0}
  &=&\sum_{j=0}^{k-1}\left|\begin{array}{cccccccc}
    A_0&A_1&\cdots&A_{j-1}&\dot{A}_j&A_{j+1}&\cdots&A_{k-1}\\
    A_1&A_2&\cdots&A_{j}&\dot{A}_{j+1}&A_{j+2}&\cdots&A_{k}\\
    \vdots&\vdots&\ddots&\vdots&\vdots&\vdots&\ddots&\vdots\\
    A_{k-1}&A_{k}&\cdots&A_{k+j-2}&\dot{A}_{k+j-1}&A_{k+j}&\cdots&A_{2k-2}
    \end{array}\right|\\
    &=&(2r+2s)G_{k-1}^{1}\\
    &&+\sum_{j=0}^{k-1}\left|\begin{array}{cccccccc}
    A_0&A_1&\cdots&A_{j-1}&[r(j+1)+2s]{A}_{j-1}&A_{j+1}&\cdots&A_{k-1}\\
    A_1&A_2&\cdots&A_{j}&[r(j+2)+2s]{A}_{j}&A_{j+2}&\cdots&A_{k}\\
    \vdots&\vdots&\ddots&\vdots&\vdots&\vdots&\ddots&\vdots\\
    A_{k-1}&A_{k}&\cdots&A_{k+j-2}&[r(j+k)+2s]{A}_{k+j-2}&A_{k+j}&\cdots&A_{2k-2}
    \end{array}\right|\\
    &=&(2r+2s)G_{k-1}^{1}\\
    &&+\sum_{j=0}^{k-1}\left|\begin{array}{cccccccc}
    A_0&A_1&\cdots&A_{j-1}&[r+2s]{A}_{j-1}&A_{j+1}&\cdots&A_{k-1}\\
    A_1&A_2&\cdots&A_{j}&[2r+2s]{A}_{j}&A_{j+2}&\cdots&A_{k}\\
    \vdots&\vdots&\ddots&\vdots&\vdots&\vdots&\ddots&\vdots\\
    A_{k-1}&A_{k}&\cdots&A_{k+j-2}&[rk+2s]{A}_{k+j-2}&A_{k+j}&\cdots&A_{2k-2}
    \end{array}\right|\\
    &=&(2r+2s)G_{k-1}^{1}+\sum_{j=0}^{k-1}\sum_{i=0}^{k-1}(-1)^{i+j}[r(i+1)+2s]A_{i+j-1}\Delta_k^0\binom{i+1}{j+1}\\
    &=&(2r+2s)G_{k-1}^{1}+\sum_{i=0}^{k-1}(-1)^i[r(i+1)+2s]\sum_{j=0}^{k-1}(-1)^{j}A_{i+j-1}\Delta_k^0\binom{i+1}{j+1}.
\end{eqnarray*}

Noting that
\[
\sum_{j=0}^{k-1}(-1)^{j}A_{j-1}\Delta_k^0\binom{1}{j+1}=G_k^{-1}
\]
and
\[
\sum_{j=0}^{k-1}(-1)^{j}A_{i+j-1}\Delta_k^0\binom{i+1}{j+1}=0 \ \ \ \text{for}\ \ \  i=1,2,\cdots,k-1,
\]
we get
\[
\dot{\Delta_k^0}=(2r+2s)G_{k-1}^1+(r+2s)G_k^{-1}.
\]
Thus the proof is completed.
\end{proof}

\textbf{Remark:} It is noted that we have introduced the term $A_{-1}$, while Beals et al. \cite{beals1999multipeakons,beals2000multipeakons} didn't. That's because it will be helpful for proving the conclusion by use of the determinant technique.

Now we shall present the proof to Theorem \ref{th:GNCH-solution} by employing  Corollaries \ref{coro:linear}, \ref{coro:sum}, \ref{coro:multiple-identity} and Lemma \ref{lemma:bilinear}, \ref{lemma:derivative2}.

\begin{proof}[The proof to Theorem \ref{th:GNCH-solution}]
Assume that $m$ is taken to be a discrete measure as \eqref{measure:m}. The GNCH equation \eqref{equ:GNCH} is reduced to
\begin{equation}\label{equ:GNCH-Hamilton}
 \dot{x}_j=(s+r)u(x_j)-r\int_{x_j}^\infty u(x,t)d x, \ \ \ \ \dot{m}_j=-[(s+r)<u_x(x_j)>+ru(x_j)]m_j.
\end{equation}
Note that
\[
\int_{x_j}^\infty u(x,t)d x=\frac{1}{2}\sum_{i=j}^Nm_i+\frac{1}{4}\sum_{i=1}^{j-1}m_ie^{2(x_i-x_j)}-\frac{1}{4}\sum_{i=j}^Nm_ie^{2(x_j-x_i)}
\]
and
\[
 <u_x(x_j)>=-\sum_{i=1}^{j-1}m_ie^{2(x_i-x_j)}+\sum_{i=j+1}^Nm_ie^{2(x_j-x_i)}.
\]
If we substitute \eqref{expression:xm2} into \eqref{equ:GNCH-Hamilton}, then \eqref{equ:GNCH-Hamilton} is written as
\begin{eqnarray*}
  &&\frac{\dot{y}_j}{1-y_j^2}=(\frac{r}{4}+\frac{s}{2})\sum_{i=1}^{j-1}g_i(1-y_i^2)\frac{(1+y_i)(1-y_j)}{(1-y_i)(1+y_j)}\\
  &&\qquad \qquad +(\frac{3r}{4}+\frac{s}{2})\sum_{i=j}^Ng_i(1-y_i^2)\frac{(1+y_j)(1-y_i)}{(1-y_j)(1+y_i)}-\frac{r}{2}\sum_{i=j}^Ng_i(1-y_i^2),\\
  &&\dot{g}_j(1-y_j^2)-2y_j\dot{y}_jg_j=g_j(1-y_j^2)\left[(\frac{r}{2}+s)\sum_{i=1}^{j-1}g_i(1-y_i^2)\frac{(1+y_i)(1-y_j)}{(1-y_i)(1+y_j)} \right.\\
  &&\qquad \qquad \qquad \qquad \left.-\frac{r}{2}g_j(1-y_j^2)\right]-(\frac{3r}{2}+s)\sum_{i=j+1}^Ng_i(1-y_i^2)\frac{(1+y_j)(1-y_i)}{(1-y_j)(1+y_i)}.\\
\end{eqnarray*}
The above system can be simplified as
\begin{eqnarray}
  \dot{y}_j
  &=&(\frac{r}{4}+\frac{s}{2})(1-y_j)^2\sum_{i=1}^{j-1}g_i(1+y_i)^2 +(\frac{3r}{4}+\frac{s}{2})(1+y_i)^2\sum_{i=j}^Ng_i(1-y_i)^2\nonumber\\
  &&-\frac{r}{2}(1-y_j^2)\sum_{i=j}^Ng_i(1-y_i^2),\label{equ:GNCH-Hamilton-y}\\
  \dot{g}_j
  &=&(\frac{r}{2}+s)g_j(1-y_j)\sum_{i=1}^{j-1}g_i(1+y_i)^2-(\frac{3r}{2}+s)g_j(1+y_j)\sum_{i=j+1}^Ng_i(1-y_i)^2\nonumber\\
  &&-rg_jy_j\sum_{i=j+1}^{N}g_i(1-y_i^2)+(\frac{r}{2}+s)g_j^2y_j(1-y_j^2)-\frac{r}{2}g_j^2(1-y_j^2)\label{equ:GNCH-Hamilton-g}
\end{eqnarray}
where the $\dot{y}_j$ in the second equation has been eliminated by inserting the first equation to the second one.

Therefore, in order to confirm Theorem \ref{th:GNCH-solution}, it is sufficient
to prove that $y_j,g_j$ satisfy the system \eqref{equ:GNCH-Hamilton-y} and \eqref{equ:GNCH-Hamilton-g}.

Substituting the expressions \eqref{expression:gy2} into Eq. \eqref{equ:GNCH-Hamilton-y}, we see that \eqref{equ:GNCH-Hamilton-y} is equivalent to
\begin{eqnarray*}
    &&\dot{\Delta}_{N-j+1}^0\Delta_{N-j}^2-\dot{\Delta}_{N-j}^2\Delta_{N-j+1}^0\\
    &=&(r+2s)(\Delta_{N-j}^2)^2\sum_{i=1}^{j-1}\frac{(\Delta_{N-i+1}^0)^2}{\Delta_{N-i+1}^1\Delta_{N-i}^1}-(r+2s)(\Delta_{N-j}^2)^2\sum_{i=1}^{j-1}\frac{\Delta_{N-i+1}^0\Delta_{N-i}^2}{\Delta_{N-i+1}^1\Delta_{N-i}^1}\\
    &&+(\frac{r}{4}+\frac{s}{2})(\Delta_{N-j}^2)^2\sum_{i=1}^{j-1}\frac{(\Delta_{N-i}^2)^2}{\Delta_{N-i+1}^1\Delta_{N-i}^1}\\
    &&+(\frac{3r}{4}+\frac{s}{2})(2\Delta_{N-j+1}^0-\Delta_{N-j}^2)^2\sum_{i=j}^N\frac{(\Delta_{N-i}^2)^2}{\Delta_{N-i+1}^1\Delta_{N-i}^1}\\
    &&-\frac{r}{2}[2\Delta_{N-j}^2\Delta_{N-j+1}^0-(\Delta_{N-j}^2)^2]\sum_{i=j}^N[2\frac{\Delta_{N-i+1}^0\Delta_{N-i}^2}{\Delta_{N-i+1}^1\Delta_{N-i}^1}-\frac{(\Delta_{N-i}^2)^2}{\Delta_{N-i+1}^1\Delta_{N-i}^1}]
  \end{eqnarray*}
where we have simplified it.

Applying Lemma \ref{lemma:derivative2} and Corollary \ref{coro:sum} and replacing $N-j$ by $k$, the above equation can be written as
\begin{eqnarray*}
  &&[(r+2s)G_{k+1}^{-1}+(2r+2s)G_{k}^1]\Delta_{k}^2-(3r+2s)G_{k}^1\Delta_{k+1}^0\\
  &=&(r+2s)(\Delta_{k}^2)^2(\frac{\Delta_{k+2}^{-1}}{\Delta_{k+1}^1}-\frac{\Delta_{N+1}^{-1}}{\Delta_{N}^1})-(r+2s)(\Delta_{k}^2)^2(\frac{G_{N}^{0}}{\Delta_{N}^1}-\frac{G_{k+1}^{0}}{\Delta_{k+1}^1})\\
  &&+(\frac{r}{4}+\frac{s}{2})\frac{1}{4}(\Delta_k^2)^2(\frac{\Delta_{N-1}^{3}}{\Delta_{N}^1}-\frac{\Delta_{k}^{3}}{\Delta_{k+1}^1})+(\frac{3r}{4}+\frac{s}{2})(2\Delta_{k+1}^0-\Delta_{k}^2)^2\frac{\Delta_{k}^{3}}{\Delta_{k+1}^1}\\
  &&-\frac{r}{2}[2\Delta_k^2\Delta_{k+1}^0-(\Delta_k^2)^2][2\frac{G_{k+1}^0}{\Delta_{k+1}^1}-\frac{\Delta_{k}^3}{\Delta_{k+1}^1}],
\end{eqnarray*}
which can be reduced to
\begin{eqnarray}
  &&[(r+2s)G_{k+1}^{-1}+(2r+2s)G_{k}^1]\Delta_{k}^2-(3r+2s)G_{k}^1\Delta_{k+1}^0\nonumber\\
  &=&(r+2s)(\Delta_{k}^2)^2(\frac{\Delta_{k+2}^{-1}}{\Delta_{k+1}^1}+\frac{G_{k+1}^{0}}{\Delta_{k+1}^1})+(r+2s)((\Delta_{k+1}^0)^2-\Delta_{k+1}^0\Delta_{k}^2)\frac{\Delta_{k}^{3}}{\Delta_{k+1}^1}\nonumber\\
  &&+r[2\Delta_{k+1}^0-\Delta_k^2][\frac{\Delta_{k+1}^0\Delta_{k}^3}{\Delta_{k+1}^1}-\frac{\Delta_k^2G_{k+1}^0}{\Delta_{k+1}^1}],\label{formula1}
\end{eqnarray}
by using the fifth relation of Corollary \ref{coro:linear}.

 We recall that
\begin{eqnarray*}
  &&(G_{k+1}^{-1}+G_{k}^1)\Delta_{k}^2-G_{k}^1\Delta_{k+1}^0\\
  &=&(\Delta_{k}^2)^2(\frac{\Delta_{k+2}^{-1}}{\Delta_{k+1}^1}+\frac{G_{k+1}^{0}}{\Delta_{k+1}^1})+((\Delta_{k+1}^0)^2-\Delta_{k+1}^0\Delta_{k}^2)\frac{\Delta_{k}^{3}}{\Delta_{k+1}^1}
\end{eqnarray*}
 holds as is described in Corollary \ref{coro:multiple-identity}. Substituting this equality into \eqref{formula1} and rearranging it, we are left to prove \begin{eqnarray*}
  r[2\Delta_{k+1}^0-\Delta_k^2][G_k^1+\frac{\Delta_{k+1}^0\Delta_{k}^3}{\Delta_{k+1}^1}-\frac{\Delta_k^2G_{k+1}^0}{\Delta_{k+1}^1}]=0.
\end{eqnarray*}
 Notice that
 \[
 G_k^1\Delta_{k+1}^1+\Delta_{k+1}^0\Delta_{k}^3-\Delta_k^2G_{k+1}^0=0
 \]
 is nothing but an identity in Lemma \ref{lemma:bilinear}.
 Thus the proof of \eqref{equ:GNCH-Hamilton-y} is completed.

Next we proceed to the proof of \eqref{equ:GNCH-Hamilton-g}.
Similar to \eqref{equ:GNCH-Hamilton-y}, \eqref{equ:GNCH-Hamilton-g} is equivalent to
\begin{eqnarray*}
  &&2\Delta_{k+1}^0\dot{\Delta}_{k+1}^0\Delta_{k+1}^1\Delta_{k}^1-(\Delta_{k+1}^0)^2(\dot\Delta_{k+1}^1\Delta_{k}^1+\Delta_{k+1}^1\dot\Delta_{k}^1)\\
  &=&(\frac{r}{2}+s)\Delta_{k+1}^1\Delta_{k}^1\Delta_{k+1}^0\Delta_{k}^2(4\sum_{i=k+1}^{N-1}\frac{(\Delta_{i+1}^0)^2}{\Delta_{i+1}^1\Delta_{i}^1}-4\sum_{i=k+1}^{N-1}\frac{\Delta_{i+1}^0\Delta_{i}^2}{\Delta_{i+1}^1\Delta_{i}^1}
    +\sum_{i=k+1}^{N-1}\frac{(\Delta_{i}^2)^2}{\Delta_{i+1}^1\Delta_{i}^1})\\
    &&-(\frac{3r}{2}+s)(\Delta_{k+1}^0)^2\Delta_{k+1}^1\Delta_{k}^1(2-\frac{\Delta_{k}^2}{\Delta_{k+1}^0})\sum_{i=0}^{k-1}\frac{(\Delta_{i}^2)^2}{\Delta_{i+1}^1\Delta_{i}^1}\\
    &&-r\Delta_{k+1}^1\Delta_{k}^1\Delta_{k+1}^0(\Delta_{k+1}^0-\Delta_{k}^2)\sum_{i=0}^{k-1}[2\frac{\Delta_{n-i+1}^0\Delta_{n-i}^2}{\Delta_{n-i+1}^1\Delta_{n-i}^1}-\frac{(\Delta_{n-i}^2)^2}{\Delta_{n-i+1}^1\Delta_{n-i}^1}]\\
    &&+(\frac{r}{2}+s)\Delta_{k+1}^0(\Delta_{k+1}^0-\Delta_{k}^2)[2\Delta_{k+1}^0\Delta_{k}^2-(\Delta_{k}^2)^2]-\frac{r}{2}(\Delta_{k+1}^0)^2[2\Delta_{k+1}^0\Delta_{k}^2-(\Delta_{k}^2)^2],
\end{eqnarray*}
where $N-j$ is replaced by $k$ for simplicity. By using Lemma \ref{lemma:derivative2} and Corollary \ref{coro:sum} and the fifth relation of Corollary \ref{coro:linear}, the above equation can be  written as
\begin{eqnarray}
  &&2\Delta_{k+1}^1\Delta_{k}^1[(r+2s)G_{k+1}^{-1}+(2r+2s)G_k^1)]-\Delta_{k+1}^0\{[(2r+2s)G_{k+1}^0-(r+s)\Delta_{k}^3]\Delta_{k}^1\nonumber\\
  &&+\Delta_{k+1}^1[(2r+2s)G_{k}^0-(r+s)\Delta_{k-1}^3]\}\nonumber\\
  &=&(\frac{r}{2}+s)\Delta_{k}^1\Delta_{k}^2(4\Delta_{k+2}^{-1}+4G_{k+1}^0-\Delta_{k}^{3})-(\frac{3r}{2}+s)\Delta_{k+1}^1\Delta_{k-1}^{3}(2\Delta_{k+1}^0-\Delta_{k}^2)\nonumber\\
    &&-r\Delta_{k+1}^1(\Delta_{k+1}^0-\Delta_{k}^2)[2G_{k}^0-\Delta_{k-1}^3]+(\frac{r}{2}+s)(\Delta_{k+1}^0-\Delta_{k}^2)[2\Delta_{k+1}^0\Delta_{k}^2-(\Delta_{k}^2)^2]\nonumber\\
    &&-\frac{r}{2}\Delta_{k+1}^0[2\Delta_{k+1}^0\Delta_{k}^2-(\Delta_{k}^2)^2].\label{formula2}
\end{eqnarray}

We recall that
\begin{eqnarray*}
&&2\Delta_{k+1}^1\Delta_{k}^1(2G_{k+1}^{-1}+2G_k^1)-\Delta_{k+1}^0[(2G_{k+1}^0-\Delta_{k}^3)\Delta_{k}^1+\Delta_{k+1}^1(2G_{k}^0-\Delta_{k-1}^3)]\\
&=&\Delta_{k}^1\Delta_{k}^2(4\Delta_{k+2}^{-1}+4G_{k+1}^0-\Delta_{k}^{3})-\Delta_{k+1}^1\Delta_{k-1}^{3}(2\Delta_{k+1}^0-\Delta_{k}^2)\\
&&+(\Delta_{k+1}^0-\Delta_{k}^2)[2\Delta_{k+1}^0\Delta_{k}^2-(\Delta_{k}^2)^2],
\end{eqnarray*}
which is the second identity in Corollary \ref{coro:multiple-identity}. Subtracting this equality multiplied by $r+s$ from \eqref{formula2}, we are left to prove
 \begin{eqnarray*}
  &&-2\Delta_{k+1}^1\Delta_{k}^1G_{k+1}^{-1}\\
  &=&\Delta_{k}^1\Delta_{k}^2(-2\Delta_{k+2}^{-1}-2G_{k+1}^0
    +\frac{1}{2}\Delta_{k}^{3})-\frac{1}{2}\Delta_{k+1}^1\Delta_{k-1}^{3}(2\Delta_{k+1}^0-\Delta_{k}^2)\\
    &&-\Delta_{k+1}^1(\Delta_{k+1}^0-\Delta_{k}^2)[2G_{k}^0-\Delta_{k-1}^3]+\frac{1}{2}\Delta_{k}^2[2\Delta_{k+1}^0\Delta_{k}^2-(\Delta_{k}^2)^2]\\
    &&-\Delta_{k+1}^0[2\Delta_{k+1}^0\Delta_{k}^2-(\Delta_{k}^2)^2].
\end{eqnarray*}
By employing identities in Lemma \ref{lemma:bilinear}
\begin{eqnarray*}
  &\Delta_{k+1}^{1}\Delta_{k-1}^{3}&=\Delta_{k}^{1}\Delta_{k}^{3}-(\Delta_{k}^{2})^2,\\
  &\Delta_{k+2}^{-1}\Delta_{k}^2&=\Delta_{k+1}^1G_{k+1}^{-1}-\Delta_{k+1}^0G_{k+1}^0
\end{eqnarray*}
and eliminating $\Delta_k^3$ and $G_{k+1}^{-1}$, what we need to prove is reduced to
\[
2(\Delta_{k+1}^0-\Delta_{k}^2)(G_{k+1}^0\Delta_{k}^1-\Delta_{k+1}^0\Delta_{k}^2-G_{k}^0\Delta_{k+1}^1)=0.
\]
This is valid
because the identity
\[
\Delta_{k+1}^0\Delta_{k}^2=G_{k+1}^0\Delta_{k}^1-G_{k}^0\Delta_{k+1}^1
\]
holds, which appears in Lemma \ref{lemma:bilinear}.

Thus \eqref{equ:GNCH-Hamilton-g} is verified and we complete the proof of Theorem \ref{th:GNCH-solution}.
\end{proof}

We end this section by a Lax pair of the GNCH equation \eqref{equ:GNCH}. It is noted that, unlike the case of Camassa-Holm equation, $z$ in the Lax pair of the GNCH equation is dependent on time $t$.
\begin{theorem}\label{th:GNCH-lax}
The GNCH equation \eqref{equ:GNCH} may be obtained by the compatibility condition of the following system
\begin{eqnarray*}
  &&f_{xx}=zmf+f,\\
  &&f_t=[(r+s)(\frac{1}{z}-u(x,t))+r\int_x^{\infty}u(y,t)dy]f_x+(\frac{r+s}{2}u_x(x,t)+\frac{r}{2}u(x,t))f,
\end{eqnarray*}
where $r,s\in\mathbb{C}$ and $z=z(t)$ is a function satisfying the differential equation
\[
z_t=-r.
\]
\end{theorem}
\begin{proof}
  The proof can be achieved by differentiating the first equation with respect to $t$ and the second twice with respect to $x$ and setting $f_{xxt}=f_{txx}$.
\end{proof}

\section{The special cases of GNCH equation \eqref{equ:GNCH}}

In this section we shall present three special cases of GNCH equation \eqref{equ:GNCH}.
\subsection{A variant of the Camassa-Holm equation}
Taking $r=0$ in the GNCH equation \eqref{equ:GNCH}, we derive a variant of the Camassa-Holm equation.
\begin{equation}\label{equ:CH-1}
  m_t+s(um)_x+smu_x=0, \ \ \ \  2m=4u-u_{xx},
\end{equation}
where $s\in\mathbb{C}$.

Actually, this equation is just the Camassa-Holm equation (under the scale transformation $x\rightarrow x, t\rightarrow st$).
%
%

 Consider the case of $r=0,s=1$ (i.e. the Camassa-Holm equation). It is worth noting that any $A_k$ in Theorem \ref{th:GNCH-solution} is a linear combination of $e^{a_j}$, which is a little different from that in Theorem \ref{th:beals}. We recall that, in Theorem \ref{th:beals}, $\lambda_j$ are required to be distinct, $a_j(0)$ are positive, and the determinant $\Delta_{k}^1$ for $1\leq k\leq N$ do not vanish so that  $x_j$ and $m_j$ exist. As for Theorem \ref{th:GNCH-solution}, we only require that $\lambda_j$ are  distinct, and the determinant $\Delta_{k}^1$ for $1\leq k\leq N$ do not vanish because $e^{a_j}$ are always positive.

\subsection{The nonisospectral Camassa-Holm equation}

When $r=1,s=0$, the GNCH equation reduces to
\begin{equation}\label{equ:CH-2}
  m_t+[(u-\int_x^\infty u(y,t)d y)m]_x+m(u_x+u)=0, \ \ \ \  2m=4u-u_{xx}.
\end{equation}
Here we call it the nonisospectral Camassa-Holm equation because its Lax pair is as below.

\begin{corollary}
The nonisospectral Camassa-Holm equation \eqref{equ:CH-2} may be obtained by the compatibility condition of the following system
\begin{eqnarray*}
  &&f_{xx}=zmf+f,\\
  &&f_t=(\frac{1}{z}-u(x,t)+\int_x^{\infty}u(y,t)dy)f_x+\frac{1}{2}(u_x(x,t)+u(x,t))f,
\end{eqnarray*}
where $z=z(t)$ is a function satisfying the differential equation
\[
z_t=-1.
\]
\end{corollary}

The following result for its N-peakon solutions holds.

\begin{corollary}\label{coro:CH2-solution}
The nonisospectral Camassa-Holm equation \eqref{equ:CH-2} admits the N-peakon solution having the form of
\begin{equation*}
  u(x,t)=\frac{1}{2}\sum_{j=1}^Nm_j(t)\exp(-2|x-x_j(t)|).
\end{equation*}
The explicit expressions for $m_j$ and $x_j$ are as follows:
  \begin{equation*}
    x_j=\frac{1}{2}\log\left(\frac{1+y_j}{1-y_j}\right), \ \ \ \ m_j=g_j(1-y_j^2)
  \end{equation*}
  with
  \begin{equation*}
    y_j=1-\frac{\Delta_{N-j}^2}{\Delta_{N-j+1}^0}, \ \ \ \ g_j=\frac{(\Delta_{N-j+1}^0)^2}{\Delta_{N-j+1}^1\Delta_{N-j}^1}.
  \end{equation*}
  Here $\Delta_k^l=\det(A_{i+j+l}(t))_{i,j=0}^{k-1}$ and the moments $A_k(t)$ are restricted by
  \begin{eqnarray*}
      A_k(t)=\sum_{j=0}^{N}(-\lambda_j)^ke^{(k+1)a_j(t)},\ \ \ a_0(t)=\log\frac{1}{2},\ \ \ \lambda_0=0
  \end{eqnarray*}
  with $\dot{a}_j(t)=-\frac{1}{\lambda_je^{a_j(t)}}, j\geq1$ and nonzero real constants $\lambda_j$.
\end{corollary}

\subsection{The mixed type Camassa-Holm equation}
Taking $r=4,s=2$ in the GNCH equation \eqref{equ:GNCH}, we have a mixed type Camassa-Holm equation
\begin{equation}\label{equ:CH-3}
  m_t+[(6u-4\int_x^\infty u(y,t)d y)m]_x+m(6u_x+4u)=0, \ \ \ \  2m=4u-u_{xx},
\end{equation}
which also admits the N-peakon solution.
\begin{corollary}\label{coro:CH3-solution}
The mixed type Camassa-Holm equation \eqref{equ:CH-3} admits the N-peakon solution having the form of
\begin{equation*}
  u(x,t)=\frac{1}{2}\sum_{j=1}^Nm_j(t)\exp(-2|x-x_j(t)|).
\end{equation*}
The explicit expressions for $m_j$ and $x_j$ are as follows:
  \begin{equation*}
    x_j=\frac{1}{2}\log\left(\frac{1+y_j}{1-y_j}\right), \ \ \ \ m_j=g_j(1-y_j^2)
  \end{equation*}
  with
  \begin{equation*}
    y_j=1-\frac{\Delta_{N-j}^2}{\Delta_{N-j+1}^0}, \ \ \ \ g_j=\frac{(\Delta_{N-j+1}^0)^2}{\Delta_{N-j+1}^1\Delta_{N-j}^1}.
  \end{equation*}
  Here $\Delta_k^l=\det(A_{i+j+l}(t))_{i,j=0}^{k-1}$ and the moments $A_k(t)$ are restricted by
  \begin{eqnarray*}
      A_k(t)=\sum_{j=0}^{N}(-\lambda_j)^ke^{(4k+8)a_j(t)},\ \ \ a_0(t)=\frac{1}{8}\log\frac{1}{2},\ \ \ \lambda_0=0
  \end{eqnarray*}
  with $\dot{a}_j(t)=-\frac{1}{\lambda_je^{4a_j(t)}}, j\geq1$ and nonzero real constants $\lambda_j$.
\end{corollary}

\begin{corollary}
The mixed type Camassa-Holm equation \eqref{equ:CH-3} may be obtained by the compatibility condition of the following system
\begin{eqnarray*}
  &&f_{xx}=zmf+f,\\
  &&f_t=(\frac{6}{z}-6u(x,t)+4\int_x^{\infty}u(y,t)dy)f_x+(3u_x(x,t)+2u(x,t))f,
\end{eqnarray*}
where $z=z(t)$ is a function satisfying the differential equation
\[
z_t=-4.
\]
\end{corollary}

At the end of this section, we illustrate the 1-peakon solution and 2-peakon solution to Eq. \eqref{equ:CH-3}, respectively.
\begin{enumerate}
  \item \textit{1-peakon.}

  Take $\lambda_1=1$ and $a_1(0)=0$.

  From Corollary \ref{coro:CH3-solution}, we easily obtain
  \[x_1=\frac{1}{2}\log[2(4t-1)^2],\ \ \ \ m_1=\frac{2}{4t-1}\]
  so that \eqref{equ:CH-3} admits the 1-peakon solution
  \[
  u(x,t)=\frac{1}{4t-1}\exp(-2|x-\frac{1}{2}\log2-\frac{1}{2}\log(4t-1)^2|).
  \]
  From this expression, we see that $t=\frac{1}{4}$ is a turning point of the 1-peakon solution. That is, the amplitude $m_1$ satisfys $m_1<0$ for $t<\frac{1}{4}$ and $m_1>0$ for $t>\frac{1}{4}$, which means the solution is from antipeakon to peakon along with $t$. A simple simulation (see Fig.1) using Matlab is given below by computing the explicit formula.

  \begin{figure}[!htp]
  \centering
  \resizebox{0.75\textwidth}{!}{
  \includegraphics{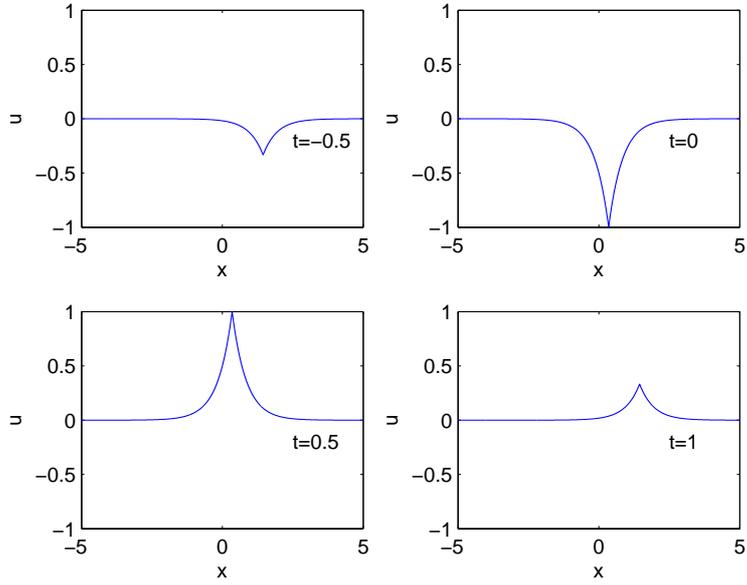}
   }
   \caption{1-peakon solution for Eq. \eqref{equ:CH-3} at time $t=-0.5,\ 0,\ 0.5,\ 1$ with the case of $\lambda_1=1,\ a_1(0)=0$}
   \label{fig:1}
\end{figure}

  \item \textit{2-peakon.}

  Take $\lambda_1=1$, $\lambda_2=-1$ and $a_1(0)=a_2(0)=0$.

  We get the 2-peakon solution (see Fig.2)
  \[
  u(x,t)=\frac{1}{2}\sum_{j=1}^2m_j(t)\exp(-2|x-x_j(t)|)
  \]
  with
  \begin{eqnarray*}
    &&x_1=\frac{1}{2}\log\frac{8(4t-1)^2(4t+1)^2}{(4t-1)^4+(4t+1)^4},\\
    &&x_2=\frac{1}{2}\log[2(4t-1)^2+2(4t+1)^2],\\
    &&m_1=\frac{2[(4t-1)^4+(4t+1)^4]}{(4t+1)(4t-1)[(4t+1)^3+(4t-1)^3]},\\
    &&m_2=\frac{2[(4t-1)^2+(4t+1)^2]}{(4t+1)^3+(4t-1)^3}
  \end{eqnarray*}
  from Corollary \ref{coro:CH3-solution}.
\begin{figure}[!htp]
  \centering
  \resizebox{0.75\textwidth}{!}{
  \includegraphics{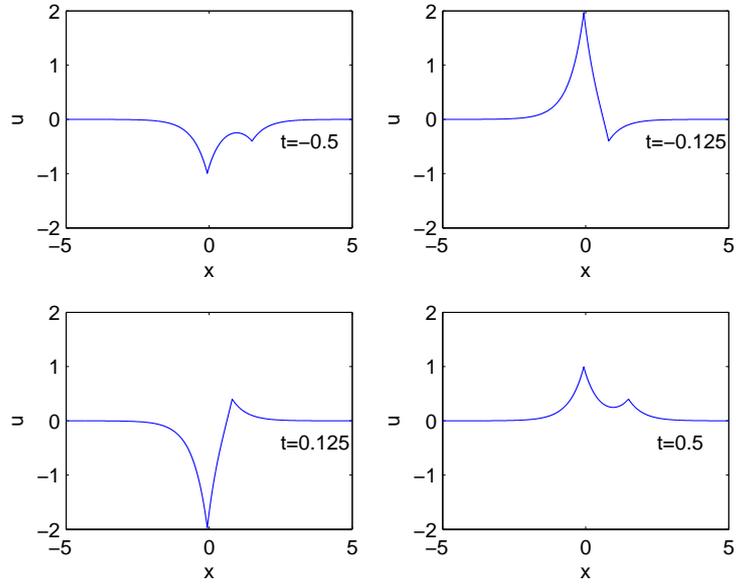}
   }
   \caption{2-peakon solution for Eq. \eqref{equ:CH-3} at time $t=-0.5,\ -0.125,\ 0.125,\ 0.5$ with the case of $\lambda_1=1,\lambda_2=-1,a_1(0)=a_2(0)=0$}
   \label{fig:2}
\end{figure}

  It is easy to work out the turning points $t=-\frac{1}{4}, 0, \frac{1}{4}$ and there hold for the amplitudes
  \begin{eqnarray*}
    &&m_1<0,m_2<0,\ \ \ \text{when}\ \ t<-\frac{1}{4},\\
    &&m_1>0,m_2<0,\ \ \ \text{when}\ \ -\frac{1}{4}<t<0,\\
    &&m_1<0,m_2>0,\ \ \ \text{when}\ \ 0<t<\frac{1}{4},\\
    &&m_1>0,m_2>0,\ \ \ \text{when}\ \ t>\frac{1}{4},
  \end{eqnarray*}
  which interpret the performance of Fig.2.
\end{enumerate}

\section{Conclusion and Discussions}
We have derived one generalized nonisospectral Camassa-Holm equation, which admits N-peakon solutions taking as similar explicit formulae as the Camassa-Holm equation. Some special cases of this system are also studied. The approach is mainly by use of the determinant technique.

After we derived our result, we accidentally found that Est\'{e}vez et al.
have investigated a nonisospectral 2+1 Camassa-Holm hierarchy in \cite{estevez2011non}, where 1+1 hierarchies are derived by using Lie symmetries reduction. The proposed equation \eqref{equ:GNCH} in this paper is different from those. And what we concern about is the N-peakon solutions.

Actually, we investigated  one kind of nonisospectral Camassa-Holm equation with N-peakon solutions. When $r,s$ in Eq. \eqref{equ:GNCH} are dependent on time $t$ instead of constants, the corresponding equation is still integrable. In other words, if we consider the more generalized determined system
\begin{eqnarray*}
  &&f_{xx}=zmf+f,\\
  &&f_t=[(r+s)(\frac{1}{z}-u(x,t))+r\int_x^{\infty}u(y,t)dy]f_x+(\frac{r+s}{2}u_x(x,t)+\frac{r}{2}u(x,t))f,
\end{eqnarray*}
where $r=r(t)$, $s=s(t)$ are two given functions dependent on time $t$ and $z=z(t)$ is a function satisfying the differential equation
\[
z_t=-r(t),
\]
then the compatibility condition leads to an equation with variable coefficients
\begin{equation*}
  m_t+[((s+r)u-r\int_x^\infty u(y,t)d y)m]_x+m[(s+r)u_x+ru]=0, \ \ \ \  2m=4u-u_{xx},
\end{equation*}
where $r=r(t)$, $s=s(t)$ are two given functions of $t$. However, we didn't derive its N-peakon solutions except the special cases with $r=0$ (the Camassa-Holm equation with variable coefficients) and $s=0$ (the nonisospectral Camassa-Holm equation with variable coefficients). It still needs further studies.

When we consider the determined system
\begin{eqnarray*}
  L_0(z)f(x,t)=0,&\ \ \ \ &L_0(z)=\partial_x^2-zm(x,t)-1,\\
  \frac{\partial f(x,t)}{\partial t}=A_0(z)f,&\ \ \ \ & A_0(z)=(\frac{1}{z}-u(x,t))\partial_x+\frac{1}{2}u_x(x,t),
\end{eqnarray*}
where $z=z(t)$ satisfy the differential equation
\[
z_t=r(t)z\]
with a given function $r(t)$. The compatibility condition yields another nonisospectral Camassa-Holm equation
\[
 m_t+(um)_x+m(u_x+r)=0, \ \ \ \  2m=4u-u_{xx}.
\]
However, this equation is not novel because it can be transformed into the Camassa-Holm equation by the variable transformation
\begin{eqnarray*}
&&X=x, \quad \qquad \qquad \qquad \qquad \qquad  T=\int^t e^{-\int^zr(y)dy}dz,\\
&&M(X,T)=m(x,t)e^{\int^tr(y)dy}, \ \ \ \ \ U(X,T)=u(x,t)e^{\int^tr(y)dy}.
\end{eqnarray*}

Therefore, it is natural to ask whether there exist more novel nonisospectral Camassa-Holm equations. Do they admit multipeakon solutions? We shall consider these problems in the future.


\section*{Acknowledgements}
This work was partially supported by the National Natural Science Foundation of China (Grant no. 11331008, 11371251), the Knowledge Innovation Program of LSEC, ICMSEC, Academy of Mathematics and Systems Science, Chinese Academy of Sciences.

\small

\end{document}